\documentclass[11pt, a4paper]{amsart}

\usepackage{amsmath}
\usepackage{latexsym}
\usepackage[all]{xy}
\usepackage{color}
\newtheorem{teorema}{Theorem}[section]
\newtheorem{definicion}[teorema]{Definition}
\include{amslatex}
\newtheorem{proposicion}[teorema]{Proposition}
\newtheorem{lema}[teorema]{Lemma}

\newtheorem{comentario}[teorema]{Remark}

\newtheorem{ejemplo}[teorema]{Example}

\numberwithin{equation}{section}

\begin{document}
\begin{title}[Weak equivalence principle in emergent quantum mechanics]
 {On the origin of the weak equivalence principle in a theory of emergent quantum mechanics}
\end{title}

\maketitle
\begin{center}
\author{Ricardo Gallego Torrom\'e\footnote{email: rigato39@gmail.com}}
\end{center}
\begin{center}
\address{Department of Mathematics\\
Faculty of Mathematics, Natural Sciences and Information Technologies\\
University of Primorska, Koper, Slovenia}
\end{center}
\begin{abstract}
 We argue that in a framework for emergent quantum mechanics, the weak equivalence principle is a consequence of concentration of measure in large dimensional spaces of $1$-Lipshitz  functions. Furthermore, as a consequence of the emergent framework and the properties that we assume for the fundamental dynamics, it is argued that gravity must be a classical, emergent interaction.
\end{abstract}
\bigskip
\section{Introduction}
The weak equivalence principle, namely, the principle of universal free falling of test particles in the presence of an isolated  gravitational field, is a cornerstone of moderns theories of gravity. It is indeed, one of the most accurately confirmed fundamental principles in physics. However, despite the principle is harmonically implemented in the mathematical structure of current gravitational theories, from general relativity to extensions and generalizations of Einstein's theory, such coherence  between the physical principle and the geometric frameworks where it is implemented does not provide an explanatory mechanism for it.

It is the main purpose of this paper to offer a derivation of the weak equivalence principle in the context of an emergent theory of gravity, a theory that the author is developing and exploring in the framework of Hamilton-Randers dynamical systems \cite{Ricardo2014}. Hamilton-Randers theory is a novel approach to the foundations of quantum theory and in some sense, with the aim to go beyond quantum theory. Such a theory builds on several general assumptions. One of the most fundamental of them is the hypothesis on the existence of a deeper level of description of physical systems than the one currently offered by the current quantum theory. In our theory, a deeper description than quantum mechanics based upon different degrees of freedom, referred as fundamental degrees of freedom, is theoretically postulated and investigated. Such fundamental degrees of freedom follow a deterministic and local dynamical law, while the quantum mechanical description of physical systems that sole assigns an element of a projective Hilbert space to each individual quantum system, is obtained as a coarse grained description from this deeper level of description. The mechanism of coarse graining is similar as how the standard thermodynamical description of a system is obtained from classical statistical mechanics. Indeed, there are several emergent theories of quantum mechanics being investigated by other authors, for instance \cite{AcostaFernandezetIsidroSantander2013,Adler,Bohm1980,Blasone2,FernandezIsidroPerea2013,Groessing2013, Singh2019 a}.

In Hamilton-Randers theory the mechanism explaining the emergence of the quantum mechanical description from a fundamental dynamics is based upon three main ingredients. The first is the novel notion of dynamics that we assume for the fundamental degrees of freedom. Such dynamical systems are two-time dynamical systems. But differently from other geometric theories of multi-time dynamics \cite{Bars2001, CraigWeinstein}, in our theory one of the time variable parameters is  necessarily emergent, as a result of the underlying fundamental, non-reversible dynamics. Such time parameters are in fact, the one used in the description of classical field theories, quantum dynamics and when pertinent, general relativistic evolution models or gravitational evolution models. At the more fundamental level of description of physical systems, the level of the dynamics of the fundamental degrees of freedom, the dynamics is non-reversible. Radically different from other approaches to the foundations of quantum theory is the two-dimensional character of time that specifies the dynamics (see chapters 2 and 3 in \cite{Ricardo06, Ricardo2014}). Let us remark that the two-time dynamics that we advocate is radically different from the one that appears in slow/fast dynamical systems \cite{Arnold}. While in classical dynamics, there is a diffeomorphic relation between the two different kind of parameters, in Hamilton-Randers theory the two times are not related by a bijection. They are truly independent parameters. Furthermore, the interpretation of the non-local character of quantum correlations is builds on the two-dimensional character of time in Hamilton-Randers dynamical systems \cite{Ricardo2014,Ricardo2017}. Therefore, our explanation of spooky distance actions is very different. This type of explanation is very different from the offered by superdeterminism \cite{Hooft2016}, even if our theory has some common points with Hooft's one, as we will see.

The second ingredient in Hamilton-Randers theory is the use of Koopman-von Neumann theory of dynamical systems. Originally in such a theory, deterministic dynamical systems are investigated using Hilbert space theory and spectral theory. In emergent quantum mechanics is essential, since it allows to consider deterministic systems from a quantum mechanical point of view. This is pivotal for the coarse grained mechanism that pass from the description using fundamental degrees of freedom to emergent quantum states (see chapters 4 and 5 in \cite{Ricardo06, Ricardo2014}). Similar techniques have been investigated by Hooft in his theory of quantum cellular automata and in earlier work \cite{Hooft,Hooft2006,Hooft2006b,Hooft2016}.

 The third fundamental ingredient in the general scheme of quantum emergent mechanics that we are pursuing is the hypothesis that the relation between the number of degrees of freedom at the fundamental scale and the number of degrees of freedom a the quantum scale are in $N$ to $1$ proportion, with $N>>1$. This is fundamental in the transition from detailed (fundamental) and coarse grained description (quantum) and in the application and consequences of a mathematical theory of analysis known as concentration of measure \cite{Gromov,MilmanSchechtman2001,Talagrand}, although to apply this theory further conditions will be required to meet, specially concerning the regularity of the evolution. It is the theory of concentration of measure that allows in Hamilton-Randers theory to provide a resolution of the measurement problem of quantum mechanics (see chapter 6 in \cite{Ricardo2014}).

It is also the application of concentration of measure that will be the key ingredient to establish the existence of a domain of the fundamental dynamics sharing a property analogous to the weak equivalence principle of gravitational theories: in such domain of the fundamental dynamics, fixed the initial conditions, an specific type of dynamical variables, namely, geometric center of masses associated to the fundamental description of quantum systems that we propose\cite{Berger2002}, follow exactly the same evolution in time, independently of the size, type and mass of the system. This is expressed in the form of exponential Chernov type bounds on the separation of the center of mass trajectories, where the large (negative) makes the differences between trajectories of different test particles practically zero for all purposes.

This interpretation of a formal weak equivalence principle of the fundamental dynamics together with other analogous properties that the fundamental dynamics should have in such domain, leads to a natural identification of such interaction with gravity.

According with the above reasoning, two main consequences follow. The first is that the weak equivalence principle is indeed an exact law of nature, in a similar way as the principles of thermodynamics are unavoidable and universal and there is no violations of it. This consequence is true for all the scales in nature up to a critical scale, where the description pass from $N$ to $1$ fundamental degrees of freedom. At that point, the consequences of concentration of measure are not valid (there is no concentration). This implies a sudden violation of the weak equivalence principle at such scale, while it is for all practical purposes preserved to any other effective scale. That is, we predict that the equivalence principle will not show up any violation until it is suddenly and abruptly violated at the fundamental scale. Second, according to our theory, the gravitational interaction is a classical, emergent interaction. It is classical because it appears in the domain where concentration is present. Such domain coincides with dynamical domain  where, according to Hamilton-Randers theory, the value of all possible observables of the system are is well defined (see for instance, chapter 6 in \cite{Ricardo2014}). It is emergent because the variables involving in the interaction that is described are average variables (geometric center of mass).
Gravity, as such, does not apply to the fundamental degrees of freedom of Hamilton-Randers theory.

The structure of this work is as follows. In section 2, we introduce the notions from Hamilton-Randers dynamical systems that are relevant for the purpose of this paper. In section 3 we discuss a technical result, the decomposition of the Hamiltonian function in matter part and a non-matter part. Section 4 discuss how concentration of measure implies a formal weak equivalence principle. Section 5 shows that Newtonian gravity holds the fundamental property need for concentration, namely, being a $1$-Lipschitz interaction. Section 6 extends the results by arguing that, by formal similarities of the dynamics of Hamilton-Randers models with the characteristic properties of gravity in relativistic theories of gravitation, that such interaction must be indeed be the gravitational interaction. In Appendix A we collect several notions of Hamilton-Randers theory, while in Appendix B we describe the notion of concentration of measure and several results of particular practical interest for our purposes.

\section{Brief introduction to Hamilton-Randers theory}\label{Introduction Hamilton-Randers spaces}

The theory of Hamilton-Randers models proposes that the quantum description of physical systems is an effective description, obtained as a coarse grained approximation of a deterministic dynamics for fundamental degrees of freedom \cite{Ricardo06,Ricardo2014}. The fundamental assumption is that any single quantum system, as it can be an isolated electron, an atomic system or any system that effectively described by a quantum model, is at deeper level described by a dynamical system with many deterministic, local degrees of freedom of a type that we have called Hamilton-Randers models. This clash with the experimental violation of Bell-like inequalities, that imposes constrains on the admissible hidden variable theories \cite{Bell}. The way these constrains are overcome is through a mechanism of interacting at the fundamental level by means of an hypothesized fundamental dynamics at the level of the fundamental degrees of freedom \cite{Ricardo06,Ricardo2014,Ricardo2017}. Essentially, the theory of emergent quantum mechanics that we propose is a new theory of dynamics and time evolution, combined with an application of an extended version of Koopman-von Neumann theory of dynamical systems \cite{Koopman1931,von Neumann,ReedSimonI}.

Without going into the many mathematical details of the theory, that the interested reader can find in  \cite{Ricardo2014}, but that we remind is still a work in progress, we resume here the essential points that are strictly relevant for the topic of this paper and that we believe will remain immune to future developments of the theory.
The fundamental assumptions in our approach to emergent quantum mechanics are the following:
\begin{enumerate}
\item  There is an underlying deterministic, local and universal dynamics beneath quantum dynamical systems,

\item The dynamics is two dimensional. Namely, the time parameters need to describe the dynamics are two dimensional pairs $(t,\tau)$ that live either in a $2$-dimensional number field $\mathbb{K}$ or in a product $\mathbb{K}_1\times \,\mathbb{K}_2$ of number fields.

\item  The number of degrees of freedom of such fundamental dynamics $(U_t,U_\tau)$ to the quantum dynamics is in relation $N:1$, with $N>>1$. Thus, for instance, a single electron or neutrino should be described by models with $N>>1$ degrees of freedom and the ration $N:1$ is a {\it measure} of the complexity of the system. The models describe the dynamics of the system in an abstract configuration tangent manifold.

\item The internal dynamics $U_t$ (also identified with what we have called fundamental dynamics and fundamental flow) is assumed to have three different regimes, that are approximately cyclically repeated: 1. Ergodic phase, 2. Contractive phase and 3. Expanding phase. The relevant phase for the probabilistic, non-local description of quantum mechanics is the ergodic phase; the relevant phase for the topic of this work is the contractive phase and the equilibrium phase. These cycles are referred as fundamental cycles.

\item A sufficient condition for this contractive phase to happen is that there is a break of the ergodicity and that the evolution operator $U_t$  is dominated by a $1$-Lipschitz evolution operator. Then by application of concentration of measure \cite{MilmanSchechtman2001,Talagrand}, one can show that for such contractive dynamics a week principle of equivalence holds universally in the contractive and equilibrium phase.
\end{enumerate}

\subsection{The geometric framework}
The configuration space in Hamilton-Randers theory are defined on products of smooth manifolds,
\begin{align}
M\cong \,\prod^N_{k=1}\,  M^k_4.
\label{manifoldM}
\end{align}
The equivalence relation $\cong $ means diffeomorphism equivalence.
Therefore, for the dynamical systems that we shall consider, the configuration manifold $\mathcal{M}$  is the tangent space of the smooth manifold $M$ and is of the form
\begin{align*}
\mathcal{M}\cong\,TM\cong \,\prod^N_{k=1}\, TM^k_4.
\end{align*}
However,  since the dynamics will be described by a Hamiltonian formalism, the relevant geometric description of the dynamics is through the co-tangent bundle $\pi:T^*TM\to TM$.

The dimension of the configuration manifold $\mathcal{M}$ is
\begin{align*}
dim(\mathcal{M})=\,dim({TM})=\, 2 \,dim (M)=\,8\,N.
 \end{align*}
For the dynamical systems that we are interested, we assume that the dimension $dim(TM)=8N$ is large for all practical purposes compared with $dim(TM_4)=8$.

As we will see, choosing the configuration space $\mathcal{M}$ as a tangent space $TM$ instead than the base manifold $M$ allows to implement geometrically second order differential equations for the coordinates of the sub-quantum molecules in a straightforward way as differential equations defining vector fields on $TM$.

 The canonical projections are the surjective maps
 \begin{align*}
 \pi_k:TM^k_4\to M^k_4.
 \end{align*}
  The vertical fiber over $x_k\in \,M^k_4$ is $\pi^{-1}_k(x_k)\,\subset TM^k_4$.
  It will also be relevant to introduce the co-tangent spaces $T^*TM^k_4$ and the projections
  \begin{align*}
  proj_k :T^*TM^k_4 \to TM^k_4,
  \end{align*}
as well as the projection
  \begin{align*}
  proj :T^*TM_4 \to TM_4.
  \end{align*}

Each manifold  $M^k_4$ is diffeomorphic to the model manifold $M_4$.
 The assumed diffeomorphisms in the theory are maps of the form
 \begin{align}
 \varphi_k: M^k_4 \to M_4,\quad k=1,...,N.
 \label{diffeomorphism conditions}
 \end{align}
 \subsection{Dynamics of the sub-quantum degrees of freedom}
The dynamical systems are of the form
\begin{align}
\begin{split}
&\dot{u}^i :=\frac{d u^i}{d t}=\,\beta^i(u),\\
 &\dot{p}_i:=\frac{d p_i}{d t}=\,-\sum^{8N}_{k=1}\,\frac{\partial \beta^k(u)}{\partial u^i}p_k,\quad i,k=1,...,8N,
\end{split}
\label{dynamicalsystem}
\end{align}
where the time derivatives are taken respect to the $t$-time parameter and determines the $U_t$ dynamics.

The on-shell constraints
\begin{align}
\dot{x}^i=\,y^i,\quad i=1,2,...,N.
\end{align}
are also imposed.  One also imposes the Randers condition,
\begin{align}
|\beta^i|<1, \,\quad i=1,...,2N,
\label{randerscondition}
\end{align}
which is equivalent to the existence of upper bounds for acceleration and velocity of the fundamental degrees of freedom.

In Hamilton-Randers theory, it is postulated that the fundamental dynamics is described by a geometric flow $U_t$ in the phase space of $M$ \cite{Ricardo2014}. The details are relevant, but still under study. However, in order to describe how the $U_t$ implies the differential equations \eqref{dynamicalsystem}, such details can be omitted here. Only it is necessary to note that, in order to accomplish with the general structure of the dynamical cycles, the Hamiltonian function can be formally express in the form
\begin{align}
H_t(u,p)=\,(1-\kappa(t,\tau))^{1/2}\,\sum^{8N}_{k=1}\,\beta^k(u)p_k,
\label{Hamiltonianfunction}
\end{align}
$\kappa(t,\tau)$ is a function that regulates the flow $U_t$. In particular, its evolution determines and defines the time parameter $\tau$ associated with the fundamental cycles.
By the assumptions of the theory during each of the fundamental cycles there is a contractive regime such that the $U_t$ dynamics is $1$-Lipschitz continuous in some operator norm. In such regime the conditions
\begin{align}
\lim_{t\to (2n+1)T}\, H=0,\quad \lim_{t\to (2n+1)T}\,\left(1-\kappa(t)\right)=0, \quad n\in\,\mathbb{Z}
\label{equilibriumlimit}
\end{align}
hold good and the dynamics is manifestly $\tau$-time diffeomorphic invariant, since the Hamiltonian function \eqref{Hamiltonianfunction} is zero or close to zero in such regime of the dynamics.
The second of these properties suggests the introduction of the domain ${\bf D}_0\subset \,T^*TM$ containing the points corresponding to the evolution at the instants
 \begin{align}
\{t=\,n\,T,\,n\in\mathbb{Z}\}.
\label{metaestabilitypoints}
  \end{align}
The open domain ${\bf D}_0$ will be called the {\it metastable equilibrium domain}.
 \subsection{Re-definition of the $t$-time parameter and $U_t$ flow}
 The $U_t$ flow has been parameterized by the conformal factor $\kappa(t,\tau)$. However, in order to obtain dynamical equations of motion \eqref{dynamicalsystem} from a Hamiltonian theory it is necessary to conveniently normalize the $t$-time parameter. Let us denote the {\it old}  time parameter, by $\tilde{t}$ and the new external time parameter by $t$.
 In order to resolve such incompatibility, we re-define
 \begin{align}
 \tilde{t}\mapsto t=\,\tilde{t}\,(1-\tilde{\kappa}(\tilde{t},\tau))^{-1}.
 \label{redefinitionoftau}
 \end{align}
  As it stands, the relation \eqref{redefinitionoftau}  is well defined, since $\kappa(\tilde{t},\tau)$ does not depend on $u\in TM$.
On the other hand, one has the differential expression
 \begin{align*}
 {d{t}}=\,(1-\tilde{\kappa}(\tilde{t},\tau){d\tilde{t}}+
 \tilde{t}{d}\left(\frac{1}{1-\tilde{\kappa}(\tilde{t},\tau)}\right).
 \end{align*}

  It is reasonable that respect to the internal time we can further impose on the equilibrium domain ${\bf D}_0$
 \begin{align*}
 \frac{d\tilde{\kappa}}{d\tilde{t}}=0,
 \end{align*}
 indicating homogeneity respect to $\tilde{t}$-time translation. We remark that this condition is imposed only on ${\bf D}_0$, that is, the above derivative is negligible in that region and zero at the points $t=\,2nT$, which is also the domain defining the $\tau$-time parameter. Then we have that on ${\bf D}_0$ the relation
 \begin{align*}
  {d{t}}=\,(1-\tilde{\kappa}(\tau)){d\tilde{t}}
  \end{align*}
  holds good.
The effective Hamiltonian that describe the dynamics the full cycle is given by
\begin{align}
\nonumber H_t(u,p) &= \sum^{8N}_{k=1}\, \beta_i(u)\,p^i, & t\neq n\,T,\,n\in\, \mathbb{Z},\\
& =\,0, &  t= n\,T,\,n\in\, \mathbb{Z}.
\label{randershamiltoniant}
\end{align}

\subsection{Koopman-von Neumann theory}
We apply Koopman-von Neumann theory \cite{Koopman1931, von Neumann} to Hamilton-Randers systems in the following way. First, we introduce the following quantization prescription, which is indeed an statement on the existence of a non-commutative Lie algebra,
\begin{align}
(u,p)\mapsto (\hat{u},\hat{p}),\quad
[\hat{u}^i,\hat{u}^j]=0,\quad [\hat{p}_i,\hat{p}_j]=0,\quad [\hat{u}^i,\hat{p}_j]=\,\delta^i_{j},\quad i,j=1,...,8N.
\label{quantumalgebra}
\end{align}
This quantization prescription is not equivalent to the standard canonical quantization of  quantum mechanics, since in our case half of the $u$-coordinates represent velocity coordinates. After quantization, all the  $\hat{u}$-velocity $\{\hat{y}^i\}^{4N}_{i=1}$ elements of the quantum algebra \eqref{quantumalgebra} commute with the $\hat{u}$-position coordinate operators $\{\hat{x}^i\}^{4N}_{i=1}$. This situation contrasts with the usual canonical quantization used in quantum theory, where the canonical position operators and the velocity operators do not commute. The fact that the dynamics of sub-quantum degrees of freedom is deterministic is not in contradiction with the uncertainty principle in quantum mechanics, since the degrees of freedom represented by points $u\in TM$ are not the same than the quantum degrees of freedom used in the quantum mechanical description of the system.

 Let us consider a linear representation of the algebra \eqref{quantumalgebra} on a vector space $\mathcal{H}_{F}$. A generation set  $\{|u\rangle =\,|x,y\rangle\}\subset \mathcal{H}_F$ of eigenstates for the $\hat{u}$ operators is also introduced. Then the  vector space $\mathcal{H}_{F}$ can be furnished with a natural scalar product in such a way that $\mathcal{H}_{F}$ with such scalar product is a pre-Hilbert space. We assume that indeed $\mathcal{H}_F$ is a Hilbert space. We can state that the operators $\{\hat{u}^\mu_k\}$ and $\{p^\mu_k\}$ are self-adjoint to respect the inner product in $\mathcal{H}_F$. Furthermore, the eigenvalues of $\{\hat{u}^\mu_k\}$ are continuous and compatible with the atlas structure of $TM$.

 The eigenvectors of $\{\hat{u}^\mu_k\}$ are such that
 \begin{align*}
 \hat{x}^\mu_k\,|x^\mu_k,y^\mu_k\rangle=\,x^\mu\,|x^\mu_k,y^\mu_k\rangle,\quad  \hat{y}^\mu_k\,|x^\mu_k,y^\mu_k\rangle=\,y^\mu\,|x^\mu_k,y^\mu_k\rangle.
 \end{align*}
In terms of the generator system of $\mathcal{H}_F$
\begin{align*}
\{|u^\mu\rangle\}^N_{k=1}\equiv\left\{|x^\mu_k,y^\mu_k\rangle,\, N=1,...,N,\, \mu=0,1,2,3\right\}
\end{align*}
a generic element of $\mathcal{H}_F$ is of the form
\begin{align}
\psi(x)=\,\sum^N_{k=1}\,\frac{1}{\sqrt{N}}\,\int_{T_xM^k_4}\,d^4z_k \,e^{\imath \varphi_k(x_k,z_k)}\,n_k(x^\mu_k,z^\mu_k)\,|x_k,z_k\rangle,
\label{quantumstates}
\end{align}
where in the following $|x_k,z_k\rangle$ stands for $|x^\mu_k,z^\mu_k\rangle$. $\vartheta_k:\mathcal{M}^k_4\to M_4$ are diffeomorphisms and we denote by $\vartheta^{-1}_k(x)=x_k$, etc...
 Note that the velocity coordinates are integrated. This is an expression of the ergodicity in one of the dynamical regimes of the $U_t$ dynamics. The ergodic theorem is only applied respect to the speed coordinates and not respect to the spacetime coordinates.

 The spacetime coordinates are labels for the degrees of freedom. The choice of the type of label is irrelevant. Hence, the theory must be invariant under diffeomorphism invariance in $M_4$. This invariant condition is totally consistent with Hamilton-Randers theory.

 The Hamiltonian operator $\widehat{H}$ is obtained from  the Hamiltonian function \eqref{Hamiltonianfunction} by application of Born-Jordan quantization prescription to the algebra \eqref{quantumalgebra}. The quantized Hamiltonian $\widehat{H}$ defines the $U_\tau$ dynamics through the corresponding Heisenberg equation. In the equilibrium domain $\widehat{H}=0$, that can be read in a weak way as $\widehat{H}\psi=0$, where $\psi$ describes a quantum state of the system.

 \section{Lower bound for the energy level of the matter Hamiltonian}\label{Lower bound for the Hamiltonian} Koopman-von Neumann formalism provides the fist step towards an unified view of classical and quantum dynamical systems. However, as it was emphasized by Hooft, to re-formulate classical dynamics from a quantum mechanical point of view, there is the fundamental problem of finding a bounded from below Hamiltonian for matter or a mechanism with the same effects. In this section we show that a natural mechanism exists in Hamilton-Randers theory that renders the Hamiltonian function for matter bounded. The treatment of density Hamiltonian functions could be traced in similar lines as here.
 
   Let us consider the decomposition of the Hamiltonian operator $\widehat{H}_t$ given by \eqref{Hamiltonianfunction} in a $1$-Lipschitz component $\widehat{H}_{Lipschitz,t}$ and a non-Lipschitz component $\widehat{H}_{matter,t}$,
   \begin{align}
   \widehat{H}_t(\hat{u},\hat{p})=\,\widehat{H}_{matter,t}(\hat{u},\hat{p})\,
   +\widehat{H}_{Lipschitz,t}(\hat{u},\hat{p}).
   \label{decompositionoftheHamiltonianHt}
   \end{align}
The {\it matter Hamiltonian} is defined in this expression by the piece of the Hamiltonian operator  which is not $1$-Lipschitz. This is consistent  with the idea that matter (including gauge interactions) is {\it quantum matter}.

 \begin{lema} \label{lemaon1lipschitz}
    Let $H_t:T^*TM\to \mathbb{R}$ be a $\mathcal{C}^2$-smooth Randers Hamiltonian function \eqref{Hamiltonianfunction}. Then there exists a compact domain $K'\subset\,T^*TM$ such that the restriction $H|_{K'}$ is $1$-Lipschitz continuous.
   \end{lema}
   \begin{proof}
   By Taylor's expansion at the point $(\xi,\chi)\in T^*TM$ up to second order  one obtains the expressions
   \begin{align*}
   H_t(u,p) & =\,H_{t}+\,\sum^{8N}_{k=1}\,\frac{\partial H_t}{\partial u^{k}}|_{({\xi},\chi)}(u^k-{\xi}^k)+\,\sum^{8N}_{k=1}\,\frac{\partial H_t}{\partial p_{k}}|_{({\xi},\chi)}(p_k-\chi_k)\\
   & +\,\sum^{8N}_{k=1}\,R_{k}\,(u^k-{\xi}^k)^2+\,\sum^{8N}_{k=1}\,Q_{k}\,(p_k-\chi_k)^2
   \end{align*}
   where the term
   \begin{align*}
   \sum^{8N}_{k=1}\,R_{k}\,(u^k-{\xi}^k)^2+\,\sum^{8N}_{k=1}\,Q_{k}\,(p_k-\chi_k)^2
   \end{align*}
   is the remaind term of the second order Taylor's expansion.
   The difference for the values of the Hamiltonian $H_t$ at two different points is given by the expression
   \begin{align*}
   &|H_t(u(1),p(1))-\,H_t(u(2),p(2))|=\Big|\sum^{8N}_{k=1}\,\frac{\partial H_t}{\partial u^{k}}|_{({\xi},\chi)}(u^k(1)-{\xi}^k)\\
   &+\,\sum^{8N}_{k=1}\,\beta^k({\chi})\,(p_k(1)-\chi_k) +\sum^{8N}_{k=1}\,R_{k}(1)\,(u^k(1)-{\xi}^k)^2+\,\sum^{8N}_{k=1}\,Q^{k}(1)\,(p_k(1)-{\xi}_k)^2\\
   & -\sum^{8N}_{k=1}\,\frac{\partial H_t}{\partial u^{k}}|_{({\xi},\chi)}(u^k(2)-{\xi}^k)-\,\sum^{8N}_{k=1}\,\beta^k({\chi})\,(p_k(2)-\chi_k)\\
   & -\sum^{8N}_{k=1}\,R_{k}(2)\,(u^k(2)-{\xi}^k)^2-\,\sum^{8N}_{k=1}\,Q^{k}(2)\,(p_k(2)-\chi_k)^2\Big|\\
   & \leq \big|\sum^{8N}_{k=1}\,\frac{\partial H_t}{\partial u^{k}}|_{({\xi},\chi)}(u^k(1)-u^k(2))\big| +\,\big|\sum^{8N}_{k=1}\,\beta^k(\chi)(p_k(1)-p_k(2))\big|\\
   & +\big|\sum^{8N}_{k=1}\,R_{k}(1)\,(u^k(1)-{\xi}^k)^2-\,R_{k}(2)\,(u^k(2)-{\xi}^k)^2\big|\\
   & +\big|\sum^{8N}_{k=1}\,Q^{k}(1)\,(p_k(1)-\chi_k)^2-\,Q^{k}(2)\,(p_k(2)-\chi_k)^2\big|.
   \end{align*}
  Due to the continuity of the second derivatives of $H_t$, for each compact set $K\subset \,T^*TM$ containing the points $1$ and $2$, there are two constants  $C_R(K)>0$ and $C_Q(K)>0$ such that $|R_{k}(1)|,|R_{k}(2)|<C_R(K)$ and $|Q_{k}(1)|,|Q_{k}(2)|<C_Q(K)$, for each $k=1,...,8N$. Moreover, as a consequence of Taylor's theorem it holds that
    \begin{align*}
    \lim_{1\to 2}\,C_Q(K)=0,\quad\lim_{1\to 2}\,C_R(K)=0,
    \end{align*}
Since $K$ is compact the last two lines in the difference $|H(u(1),p(1))-\,H(u(2),p(2))|$ can be rewritten as
\begin{align*}
&\big|\sum^{8N}_{k=1}\,R_{k}(1)\,(u^k(1)-{\xi}^k)^2-\,R_{k}(2)\,(u^k(2)-{\xi}^k)^2\big|
\,\leq C_R(K)\big|\sum^{8N}_{k=1}(u^k(1)-u^k(2))^2\big|\\
   & \big|\sum^{8N}_{k=1}\,Q^{k}(1)\,(p_k(1)-\chi_k)^2-\,Q^{k}(2)\,(p_k(2)-\chi_k)^2\big|\,\leq
  C_Q(K) \big|\sum^{8N}_{k=1}\,(p_k(1)-p_k(2))^2\big|.
\end{align*}
The constants $C_Q(K)$ and $C_R(K)$ can be taken finite on $K$.
    Furthermore,  by further restricting the domain where the points $1$ and $2$ are to be included in a smaller compact set $\tilde{K}$, one can write the following relations,
    \begin{align}
   C_R(\tilde{K})|(u^k(1)-u^k(2))|\leq 1/2,\,\quad C_Q(\tilde{K})|(p_k(1)-p_k(2))|\leq 1/2.
\label{boundoftheconstant}
    \end{align}
Let us consider the  further restriction on the compact set ${K}'\subset \tilde{K}\,\subset\, T^*TM$ such that for each $({\xi},\chi)\in\,{K}'$
\begin{align}
\big|\frac{\partial H_t}{\partial u^{k}}|_{({\xi},\chi)}\big|\leq C_U,\,k=1,....,4N
\label{boundinacceleration}
\end{align}
holds good for some constant $C_U$. Also, on ${K}'$ it must hold that
\begin{align*}
& C_R(K)\big|\sum^{8N}_{k=1}(u^k(1)-u^k(2))^2\big|+\,  C_Q(K) \big|\sum^{8N}_{k=1}\,(p_k(1)-p_k(2))^2\big|\\
& \leq \,1/2\,\sum^{8N}_{k=1}\big|(u^k(1)-u^k(2))\big|+ 1/2\,\sum^{8N}_{k=1}\,\big|(p_k(1)-p_k(2))\big|.
\end{align*}
Moreover, the factors $|\beta^i|$ are bounded as a consequence of Randers condition \eqref{randerscondition}.
As a consequences of these inequalities and relations we have that on the compact set $K'\subset \,T^+TM$ the relation
\begin{align*}
&|H(u(1),p(1))-\,H(u(2),p(2))|\big|_{{K}'}\leq  \,\tilde{C}_U\,\Big(\sum^{8N}_{k=1}\big|(u^k(1)-u^k(2))\big|\\
&+\,\sum^{8N}_{k=1}\,\,\big|\,(p_k(1)-p_k(2))\big|\Big) + 1/2\,\sum^{8N}_{k=1}\big|(u^k(1)-u^k(2))\big|+ 1/2\,\sum^{8N}_{k=1}\,\big|(p_k(1)-p_k(2))\big|
\end{align*}
with $\tilde{C}_U=\,\max \{C_U,1\}$ holds good.
This proves that $H|_{K'}$ is a Lipschitz function, with Lipschitz constant $\widetilde{\tilde{C}}_U =\max\{\frac{1}{2},\tilde{C}_U\}$, which is necessarily finite. Now we can redefine the Hamiltonian dividing by $\widetilde{\tilde{C}}_U$, which is a constant larger than $1$. This operation is equivalent to redefine the vector field $\beta\in \,\Gamma TTM$. Such operation  does not change the equations of motion and the Randers condition \eqref{randerscondition} on the compact set $K'$. Then we  obtain a $1$-Lipschitz Hamiltonian function on $K'$, as a restriction and constant re-scaling of the original Hamiltonian $H_t$.
    \end{proof}

For the hypothesized Hamilton-Randers dynamics, the compact domain $K'$ is not empty. In the metaestable domain ${\bf D}_0$, the Hamiltonian $H_t$ is equivalent to zero. Therefore, it is reasonable to think that in such domain $H_t$ is Lipschitz, thus providing an example where $K'$ can be contained (${\bf D}_0$ is not necessarily compact). This argument also shows that for the case of Hamilton-Randers dynamics $K'$ is not a discrete set.

Extensions from $K'$ to the whole phase space $T^*TM$ can be constructed as follows. Consider the Sasaki metric on $T^*TM$ of the Hamilton-Randers structure (see \eqref{DefinicionHR} in Appendix \ref{Hamilton-Randers structures}). For every observer $W$ one can associate by canonical methods a Finsler metric on $T^*TM$ and then an asymmetric distance function
   \begin{align*}
   \varrho_S:T^*TM\times T^*TM\to \mathbb{R}.
   \end{align*}
   Let us consider the projection on $K'$
   \begin{align}
   \pi_{K'}: T^*TM\to K',\,(u,p)\mapsto (\bar{u},\bar{p}),
   \end{align}
  where $(\bar{u},\bar{p})$ is defined by the condition that the distance from $(u,p)$ to $K'$ is achieved at $(\bar{u},\bar{p})$ in the boundary $\partial K'$. Then one defines the {\it radial decomposition} of $H_t$ by the expression
   \begin{align}
   H_t(u,p)=\,R\big(\varrho_S((u,p),(\bar{u},\bar{p}))\big)\,H_t(\bar{u},\bar{p})+\,\delta H_t(u,p).
   \label{decompositionofH}
   \end{align}
   The positive function $R\big(\varrho_S((u,p),(\bar{u},\bar{p}))\big)$ is such that $H_t(\bar{u},\bar{p}$, which is the first piece of the re-scaled Hamiltonian, is $1$-Lipschitz by lemma \ref{lemaon1lipschitz}. The second contribution is not $1$-Lipschitz. By assumption, $\delta H_t(u,p)$ is identified with the matter Hamiltonian $H_{matter}$,
    \begin{align}
    H_{matter, t}(u,p):=\,\delta H_t(u,p).
    \label{matterhamiltonian}
    \end{align}
With these redefinitions we obtain the following
   \begin{teorema}
    Every Hamiltonian  \eqref{Hamiltonianfunction} admits a normalization such that the decomposition \eqref{decompositionoftheHamiltonianHt} holds globally on $T^*TM$.
   \label{radialdecomposition}
   \end{teorema}
   \begin{proof}
   One can perform the normalization
   \begin{align*}
   H_t(u,p)\to \,&\frac{1}{R\big(\varrho_S((u,p),(\bar{u},\bar{p}))\big)}\,H_t(u,p)\\
& = H_t(\bar{u}(u),\bar{p}(u))+\, \frac{1}{R\big(\varrho_S((u,p),(\bar{u}(u),\bar{p}(u)))\big)} \delta H_t(u,p).
 \end{align*}
 The first term is $1$-Lipschitz in $T^*TM$, since $H_t(\bar{u}(u),\bar{p}(u))$ is $1$-Lipschitz on $K'$, while the second term is not $1$-Lipschitz continuous.
   \end{proof}
   The uniqueness of this construction depends upon the uniqueness of the compact set $K'$, the uniqueness of the relation $(u,p)\mapsto (\bar{u},\bar{p})$. One can consider the maximal set $K'$, but in general the construction is not unique.
\subsection{Hamiltonian constrain}
   From the properties of the $U_t$ flow it follows, after quantization, that
   \begin{align}
   \lim_{t\to (2n+1)T} \big(\widehat{H}_{matter,t}\,+\widehat{H}_{Lipschitz,t}\big)|\psi\rangle=0
   \label{Hamiltonian constrain}
   \end{align}
   must hold for each $|\psi\rangle\in\,\mathcal{H}_{F}$. This relation is identified with the quantum version of the Hamiltonian constrain in time reparametrization invariance in gravitational theories.
   However, each of the individual terms in this relation can be different from zero in the metastable domain ${\bf D}_0$,
   \begin{align*}
   \lim_{t\to (2n+1)T} \widehat{H}_{matter,t}|\psi\rangle\neq \,0
    \end{align*}
    and
    \begin{align*}
    \lim_{t\to (2n+1)T}\widehat{H}_{Lipschitz,t}|\psi\rangle\neq\,0.
   \end{align*}
   This implies that in order to have the metastable equilibrium point at the instant $t=(2n+1)T$, in addition with the matter Hamiltonian \eqref{matterhamiltonian}, an additional piece of dynamical variables whose described by the Hamiltonian $\widehat{H}_{Lipschitz,t}$ is needed. On the other hand,
if we assume that the matter Hamiltonian \eqref{matterhamiltonian} is positive definite when acting on physical states (weakly positive), then the $1$-Lipschitz Hamiltonian should be negative when acting on physical states. Hence for Hamilton-Randers dynamical models the positiveness of the matter Hamiltonian is extended to all $t\in [0,(2n+1)T]$.
This implies the consistency of the positiveness of the energy level for the quantum Hamiltonian for matter \eqref{matterhamiltonian} in the whole process of the $U_t$-evolution.
 On the other hand, $\widehat{H}_{Lipschitz,t}$ is related with the classical gravitational interaction.

One can reverse this argument. If $\widehat{H}_{Lipschitz,t}$ is weakly negative definite, then $\widehat{H}_{matter,t}$ must be weakly positive definite. This property is related with the analogous property of gravitational interaction, as we discuss below.

\section{Emergence of the weak equivalence principle in Hamilton-Randers theory}
We organize the contents of this section in two parts. In the first part, preparatory material will formalize the theory of center of mass coordinates for Hamilton-Randers systems and relate them with macroscopic free falling coordinates. In the second part, we will apply concentration of measure to such coordinate functions, deducing in a formal way the weak equivalence principle.
\subsection{Preliminary considerations}
Let us consider a physical system $\mathcal{S}$ that can be thought as composed by two sub-systems $A$ and $B$.
     We denote by  $X^\mu(\mathcal{S}_i),\,i\equiv \mathcal{S},A,B$ the {\it macroscopic observable} coordinates associated to the system $\mathcal{S}_i$, that is, the value of the coordinates that could be associated when local coordinates are assigned by a classical observer to each system $\mathcal{S},A,B$ by means of a measurements or by means of a given theoretical model.
Then we adopt the following
\bigskip

{\bf Assumption I}. The functions
\begin{align*}
\begin{split}
& X^\mu(\mathcal{S}_i):T^*TM\times \mathbb{R}\to \mathbb{R},\\& (u(1),...,u(N),p(1),...,p(N),t)\mapsto X^\mu(u(1),...,u(N),p(1),...,p(N),t)
\end{split}
\end{align*}
are smooth.

\bigskip

Under this hypothesis, that can be taken at the present stage of the theory as working hypothesis,
we will show that there is a natural collection of functions $X^\mu(t)$ such that
\begin{align*}
\tau\mapsto X^\mu(\tau):=\,X^\mu(\tau T)
\end{align*}
  are $1$-Lipschitz in $\tau$-time parameter, where $\tau=2n+1, \,n\in\mathbb{Z}$.

\bigskip

{\bf Assumption II}. In the metastable domain ${\bf D}_0$ the functions
\begin{align*}
\left\{u^\mu_k(t),\frac{d u^\rho_k}{dt},\,\mu=1,2,3,4;\, k=1,...,N\right\}
\end{align*}
 are  $\mathcal{C}^1$-regular functions.
\bigskip

Let us consider two  subsystems $A$ and $B$ of the full system under consideration $\mathcal{S}$. The sub-systems $A$, $B$ are embedded in $\mathcal{S}$ such that
\begin{align}
\mathcal{S}=\,A\sqcup B,
 \end{align}
 for an hypothetically  well defined union operation $\sqcup$ for systems composed by sub-quantum molecules. Let us consider local coordinate systems such that the identification
   \begin{align}
   A\equiv (u_1(\tau),...,u_{N_A}(\tau),0,...,0) \quad \textrm{and}\quad B\equiv (0,...,0,v_1(\tau),...,v_{ N_B}(\tau)),
   \label{embedding A,BtoS}
   \end{align}
   with $N=\,N_A+N_B,N_A,N_B\gg 1$ holds good.
The whole system $\mathcal{S}$ can be represented in local coordinates as
   \begin{align*}
   \mathcal{S}\equiv (u_1(t),...,u_{ N_A}(t),v_1(t),...,v_{ N_B}(t)).
   \end{align*}
    By the action of the diffeomorphisms $\varphi_k:M^k_4\to M_4$, one can consider the world lines of the sub-quantum molecules on $M_4$ at each constant value of $t$ modulo $2T$. In the particular case of metaestable equilibrium points $\{\tau\equiv \,t(n)=(2n+1)T,\,n\in\,\mathbb{Z}\}$ we have a set of (discrete) world lines in $M_4$, showing that the functions $\{X^{\mu}\}^4_{\mu=1}$ characterize the average presence of sub-quantum world lines at a given point of $M_4$. Therefore, it is reasonable to define the {\it observable coordinates of the system} by the expression
    \begin{align}
   \tilde{X}^\mu_i(\tau)\equiv \tilde{X}^\mu_i(t(n))=\frac{1}{N}\,\lim_{t\to (2n+1)T}\sum^{N_i}_{k_i=1}\,\varphi^\mu_{k_i}(x_{k_i}(t)),\quad\,i=A,B,\mathcal{S},\,\mu=0,1,2,3,
    \label{definicionofXmu}
    \end{align}
    where here $\varphi^\mu_{k_i}$ are local coordinates on $M_4$, defined after the action of the diffeomorphism $\varphi_{k_i}$. Note that the normalization factor $1/N$ is the same for all the systems $i=A,B,\mathcal{S}$. This means that we are considering systems that eventually are sub-systems (proper or improper) of a larger sub-system $\mathcal{S}$. This formal constraint is however harmless for general purposes by the embedding \eqref{embedding A,BtoS}. Actually, we can suppress the factor $N$, as long as we keep track in all the expressions below of the equivalent normalization criteria.

Then by the embedding \eqref{embedding A,BtoS},
      \begin{align}
    {X}^\mu(\tau)=\frac{1}{N}\,\lim_{t\to (2n+1)T}\sum^{N}_{k=1}\,\varphi^\mu_{k_i}(x_{k_i}(t)),\quad\,i=A,B,\mathcal{S}.
    \label{definicionofXmu2}
    \end{align}
    This function is $1$-Lipschitz for a convenient choice of the diffeomorphisms $\varphi_k$, that could include local scale redefinitions depending upon $\tau$. This condition is consistent with the assumption that on the domain ${\bf D}_0$ the dominant part of the Hamiltonian is the $1$-Lipschitz Hamiltonian by assumption.

The {\it mean} function coordinate $M^\mu(\tau)$ is assumed to be equal to the mean of the probability distribution $\mu_P$,
\begin{align}
M^\mu(\mathcal{S}(\tau))=\,\frac{1}{N}\,\lim_{t\to (2n+1)T}\,\sum^{N}_{k=1}\,\mu_P(k)(t)\varphi^\mu_{k_i}(x_{k_i}(t))
\label{equationforintermsofdegreesoffreedom}
\end{align}
The mean $M^\mu$ does not depend on the system $i=A,B,\mathcal{S}$ and only depends upon the distribution of probability or measure $\mu_P$. This probability distribution of sub-quantum degrees of freedom depends only upon the macroscopic preparation of the system and not of the system itself: by definition, the probability distribution of the $N$ sub-quantum molecules does not depend upon the particular configurations that they can have. Therefore, we make a third assumption,
\bigskip

{\bf Assumption III}. The mean coordinate functions $M^\mu$ only depends on the preparatory macroscopic conditions.
\bigskip

It is under assumption III  together with the fact that the configuration spacetime is large dimensional and the identification of the evolution in ${\bf D}_0$ with a Lipschitz type evolution, that one can  apply effectively concentration of measure and with the required regularity conditions of the $U_t$-interaction, leading to a formal {\it classical weak equivalence principle}. The mean coordinate world line $\tau\mapsto M^\mu(\tau)$ serves as a {\it guide} to the motion of macroscopic variables.

From now on, we will consider the continuous limit where $d\tau <<\,\tau$. In this case, $\tau$ can be seen as a continuous real parameter $\tau\in\,I\subset\,\mathbb{R}$.
It is important to remark that in general, $\tau\mapsto M^\mu(\tau)$ is not necessarily continuous from the point of view of the spacetime description.
\begin{definicion}
Given a Hamilton-Randers system, we say that is in {\it free quantum evolution} if the condition on the probability distribution
 \begin{align}
 \frac{d}{d\tau}\mu_P(k)(t,\tau)=0,\quad k=1,...,N
 \label{conservation}
 \end{align}
holds good.
\label{freessytem}
\end{definicion}
Thus an interaction in Hamilton-Randers theory is associated with the exchange of sub-quantum degrees of freedom with the environment. Although there is still a long way, this notion suggests the possibility to describe interactions by using quantum field theory. This conjecture is also supported by the fact that Hamilton-Randers models are causal, in the sense that there is a maximal speed for physical degrees of freedom.

By the concentration property \eqref{concentration2} of the $U_t$ dynamics in the Lipschitz dynamical regime ${\bf D}_0$, the $\tau$-evolution of the coordinates $\tilde{X}^\mu(\mathcal{S}(\tau))$, $\tilde{X}^\mu(A(\tau))$ and $\tilde{X}^\mu(B(\tau))$ that have the same initial conditions  differ between each other after the dynamics at $\tau$-time such that
   \begin{align}
  \mu_P\left(\frac{1}{\sigma_{\tilde{X}^\mu}}\,|\tilde{X}^\mu(\mathcal{S}_i(\tau))-M^\mu(\mathcal{S}(\tau))|>\rho\right)_{t\to (2n+1)T}\sim C_1\exp \left(-\,C_2 \frac{\rho^2}{2\,\rho^2_p}\right),
\label{generalconcentrationofmeasure}
\end{align}
$\mu=1,2,3,4,\,i\equiv A,B,\mathcal{S}$ holds.
The constants $C_1,C_2$ are of order $1$, where $C_2$ depends on the dimension of the spacetime $M_4$. $\rho_p$ is a constant independent of the system $i=A,B,\mathcal{S}$. Note that there is no dependence on the $t$-time parameter, since we are considering these expressions in the limit $t\to (2n+1)T$.

  \subsection{Proof from "first principles" of the weak equivalence principle}
  \begin{definicion} A test particle system is described by a Hamilton-Randers system such that the $\tau$-evolution of the center of mass coordinates $\tilde{X}^\mu$  are determined by the initial conditions $\left(\tilde{X}^\mu(\tau=0),\frac{d \tilde{X}^\mu(\tau=0)}{d\tau}\right)$ and the external field.
  \label{definicionoftestparticle}
  \end{definicion}
  Moreover, we assume  that the condition
  \begin{align}
\frac{\rho}{\rho_P}\sim\, N,
\label{scales}
\end{align}
  holds for Hamilton-Randers dynamical systems. This condition indicates an equipartition of the contribution to the difference in coordinates $|\tilde{X}^\mu(\mathcal{S}_i(\tau))-M^\mu(\mathcal{S}(\tau))|$ by each fundamental degree of freedom $k=1,...,N$. One can argue that this condition holds when the system is in state of equilibrium. In Hamilton-Randers theory, these situations happen in the metaestable domain, that is, exactly when the dynamics is also $1$-Lipschitz (\cite{Ricardo2014}, chapter 6). It is under the condition \eqref{scales} that the consequences of concentration in Hamilton-Randers theory take the stronger form.
\begin{proposicion}
Let $\mathcal{S}_i,\,i=1,2,3$ be Hamilton-Randers systems with $N\gg 1$ associated to free test particles. Assume that the dynamical regime is such that the equipartition condition \eqref{scales} holds good.
 Then the observable macroscopic coordinates $\tilde{X}^\mu(\tau)$  do not depend on the system $\mathcal{S}_i$ at each time $\tau$.
 \label{proposiciononweakequivalenceprinciple}
\end{proposicion}
\begin{proof}
 The coordinate functions $\tilde{X}^\mu(\tau)$ are $1$-Lipschitz in the metaestable domain $t\to (2n+1)T$. Then we can apply the concentration of measure \eqref{generalconcentrationofmeasure}. Under the assumption \eqref{scales}, the observable coordinates $\{\tilde{X}^\mu\}^4_{\mu=1}$ moves following the common $M^\mu(\tau)$ coordinates with an error bounded by $\exp(-C_2 N^2)$. Since the system is in free evolution, the condition \eqref{conservation} holds, the median coordinates $M^\mu(\tau)$ follow an ordinary differential equation.
 \end{proof}
 \begin{comentario}
 If the relation \eqref{scales} fails to be implemented in some circumstances of the $1$-Lipschitz dynamical regime, still the bound in the measure for the points where $|\tilde{X}^\mu(\mathcal{S}_i(\tau))-M^\mu(\mathcal{S}(\tau))|$ is larger that $\rho$ is a Gaussian function. 
 \end{comentario}
 Our application of concentration, as in many others applications of concentration in geometry, probability theory and statistical mechanics (see for instance \cite{Gromov}, pg. 144 to 151) is based upon two fundamental conditions:
 \begin{itemize}
 \item The functions under consideration are $1$-Lipschitz functions in the metric of the space,

 \item The spaces under consideration are high dimensional spaces.
 \end{itemize}
 The failure of any of these conditions implies the failure in the strength of the consequences of concentration in the way used in the following proof of a formal weak equivalence principle. Indeed, for the argument below the condition \eqref{scales} with $N>>1$ is essential. On the other hand, if the dimension of the spaces is small, the concentration inequalities will not have any strong in the consequences of concentration, while if the conditions of $1$-Lipshitz fails, the concentration of measure inequalities will simply not hold.

 Therefore, in the subset of metaestable domain for $t= (2n+1)T$ there is a strong concentration for the value the functions $\{\tilde{X}^\mu(\tau)\}^4_{\mu=1}$  around the mean $\{M^\mu(\tau)\}^4_{\mu=1}$. This universality is after fixing the initial conditions of the mean coordinates $M^\mu(\tau)$, which is equivalent to fix the initial conditions for $\{u^\mu_k\}^{N}_{k=1}$. This result can be interpreted as the  weak equivalence principle.

Despite being applied to classical trajectories only, the derivation of the weak equivalence principle offered along these lines implies that theoretically the weak equivalence principle should be an almost exact law of Nature, only broken at scales comparable with the fundamental scale. This is because the weak equivalence principle just derived is valid up to an error of order $\exp(-C_2 N^2)$. However, the principle breaks down abruptly when the system described is  composed by few fundamental degrees of freedom.

Although the current state of the theory cannot determine the constant $C_2$, this does not invalidate our argument. Also, in order to obtain firm bounds on the exponential, it is necessary to know the values of $N$ for particular systems. At the moment, we assume that $N$ is large enough to effectively apply the concentration of measure argument.

\section{On the $1$-Lipschitz continuous character of gravity}
In the derivation of the formal weak equivalence principle we have made use of the assumption that the dynamics is, for the corresponding regime, $1$-Lipschitz. Although this could be formally the case, due to the formal decomposition discussed in section \ref{Lower bound for the Hamiltonian}, in order to close the argument, we should show that gravity is a $1$-Lipschitz interaction. This is in general a difficult issue, but we can argue that this should be the case considering the simplest description of gravity, namely, Newtonian gravity.

Let us consider the newtonian gravitational force between a massive point particle  with mass $m$ by a massive point particle with mass $M$  located at the origin of coordinates,
   \begin{align}
   F(\vec{x})=\,-G\,\frac{m\, M}{r^2}, \quad \vec{x}\in\,\mathbb{R}^3
   \label{Newtonlaw}
   \end{align}
  and $r=\,|\vec{x}\|$ the distance to the origin in $\mathbb{R}^3$ of  the point $\vec{x}$.  The newtonian gravitational potential $V(\vec{x})$ lives on the collection of Euclidean spheres
  \begin{align*}
   \hat{S}^2:=\{S^2(r),\,r\in\,(0,+\infty)\},
   \end{align*}
 where the expression $\|r_1-r_2\|$ defines a norm function between different spheres. Moreover, to compare different lengths or different mechanical forces, it is useful to consider dimensionless expressions, for which we need reference scales.

  The Planck scale force provides a natural unit, respect to which  we can compare any other scale.
Let us consider the comparison of forces
\begin{align*}
\frac{|F(\vec{x}_2)-F(\vec{x}_1)|}{{F}_P}=\,\alpha \,\frac{\|r_2-r_1\|}{\l_P},
\end{align*}
where $F_P$ is the Planck force and $\l_P$ is the Planck length. After some algebraic manipulations, one finds an expression for the coefficient $\alpha$. In  the case of Newton law of universal gravitation \eqref{Newtonlaw}, $\alpha$ is given by the expression
\begin{align*}
\alpha=\,\l_P\,\frac{1}{c^4} \, G^2 \,m\,M\,\frac{1}{r^2_2\,r^2_1}\,\|r_2+\,r_1 \|.
\end{align*}
In order to simplify the argument, let us  consider $m=M$. Furthermore, although the case $r_2=r_1$ is singular, in order to work in a fixed scale, we consider a relation $r_1= \lambda\, r_2$ with $\lambda\sim 1$ constant. Then one obtains a compact expression for $\alpha$,
\begin{align}
\alpha=\, \frac{1+\lambda}{\lambda^3}\,\frac{D}{D_p}\,\frac{E}{E_P},
\label{Lipschitzconstantforgravity}
\end{align}
where $D=\,{m}/{r^3}$ is a characteristic density of the system, $E=mc^2$, $D_P$ is the Planck density and $E_p$ is the Planck energy. It follows from the expression \eqref{Lipschitzconstantforgravity} that for scales of the standard model, atomic physics systems or macroscopic systems,  $\alpha\ll 1$. Moreover, $\alpha$ is bounded by $1$. The bound is saturated at the Planck scale. This shows that at macroscopic or quantum physical scales, gravity is $1$-Lipschitz. This is because the relative weakness of the gravitational interaction compared with the interactions of the Standard Model of particles.

  \section{Conclusion: On the emergent origin of the gravitational interaction in Hamilton-Randers theory}\label{euristicargumentforemergenceofgravity}
   Bringing together  the previous characteristics for the $1$-Lipschitz domain of the $U_t$ flow, we find  the following general features:
   \begin{itemize}
   \item Since the constraint \eqref{equilibriumlimit} holds good, the dynamical $U_t$ flow in the domain ${\bf D}_0$ is compatible with the time reparametrization invariance  of general relativity.

   \item The weak equivalence principle for the observable coordinates $X^\mu(S(\tau))$ holds good in the metastable domain ${\bf D}_0$.

   \item The dynamical $U_t$ flow in the domain ${\bf D}_0$ determines a classical interaction, since the metastable domain ${\bf D}_0$ of the fundamental dynamics corresponds to the domain where all the possible observables of the systems have well defined values.

   \item There is a local maximal speed for the fundamental degrees of freedom of any Hamilton-Randers dynamical system. Invariance under a local relativity group invariance holds. This local relativity group is by construction the Poincar\'e group.

   \end{itemize}
   Furthermore, we have found the following an additional restriction,
   \begin{itemize}

   \item The $U_t$ interaction in the $1$-Lipschitz domain must be compatible with the existence of an universal maximal proper acceleration.

   \end{itemize}
    In view of the formal similarity of these properties with the analogous properties of the current mathematical description of the gravitational interaction, one is naturally lead to the following conclusion,
\begin{center}
   {\it In the metastable domain the $1$-Lipschitz dynamics associated with ${H}_{Lipshitz,t=(2n+1)T}$ is a form of gravitational interaction compatible with the existence of an universal maximal acceleration.}
\end{center}

That gravity could be intrinsically involved in the collapse of the wave functions is an idea that appears in several modern approaches to the measurement problem \cite{Diosi, GhirardiGrassiRimini, Penrose}. However, there are fundamental differences between the models described here and spontaneous collapse models or collapse models induced by large mass measurement devices. Furthermore, our argument shows that gravitation is an emergent interaction, not applicable to the fundamental degrees of freedom. How this consequence  marries Hamilton-Randers theory with another fundamental property of current theories of gravitation, namely, the general covariance or absence of back-ground geometric structures, is another problem that deserves further investigation and separate exposition.

Therefore, from the point of view of Hamilton-Randers theory, gravity is a classical interaction. Furthermore, there must exist essential differences between gravitational models compatible with Hamilton-Randers theory and Einstein's general relativity, since our theory includes an universal  maximal proper  acceleration. It is also interesting the possibility that a generalization of Einstein gravity in the frameworks of metrics with a maximal acceleration compatible with the weak equivalence principle  could lead to a {\it classical resolution of curvature singularities} \cite{CaianielloGasperiniScarpetta, Ricardo2020}.

\appendix
\section{Hamilton-Randers structures}\label{Hamilton-Randers structures}
{\bf Notion of pseudo-Randers space}. A natural way to introduce a non-reversible dynamics is to consider a non-reversible perturbation to a reversible one. This is exactly the fundamental ingredient of the concept of Randers spaces \cite{Randers}. We will introduce a general notion of Randers space in the following paragraphs.
\begin{definicion}
In the category of Finsler spaces with Euclidean signature, a Randers structure defined on the manifold $\widetilde{M}$ is a Finsler structure such that the associated Finsler function is of the form
\begin{align*}
F^*:T\widetilde{M}\to \mathbb{R},\quad  (u,z)\mapsto \alpha^*(u,z)+\beta^*(u,z),
\end{align*}
such that The condition
 \begin{align}
 \alpha^*(\beta^*,\beta^*)<1
 \label{randerscondition2}
 \end{align}
 must be satisfied.
$\alpha^*(u,z)$ is the Riemannian norm of $z\in \,T_u\widetilde{M}$ determined by a Riemannian metric $\eta^*$, while $\beta^*(u,z)$ is the result of the action $1$-form $\beta^*\in \,\Gamma T^*\widetilde{M}$ on  $z$.
 \label{Randers space}
 \end{definicion}

 This condition \eqref{randerscondition} implies the non-degeneracy  and the positiveness of the associated fundamental tensor
  \begin{align}
g^{ij}(u,p)=\,\frac{1}{2}\,\frac{\partial^2 F^2(u,p)}{\partial p_i\partial p_j}.
\label{fundamentaltensor}
\end{align}
The proofs of these properties are indicated for instance  in \cite{BaoChernShen}. The non-degeneracy of the fundamental tensor is an analytical requirement for the construction of associated connections and also, for the existence of geodesics as local extremals of an action or energy functional \cite{Whitehead1932}.

 We now consider the analogous of a Randers structure in the category of generalized Hamiltonian functions on the configuration space $TM$ whose fundamental tensors \eqref{fundamentaltensor} are non-degenerate and have indefinite signature. In this case the domain of definition of the Hamiltonian function $F$ should be restricted, since it is not possible to have a well defined Hamilton-Randers function on the whole cotangent space $T^*TM$. This is because $\eta$ is a pseudo-Riemannian metric and it can take negative values on certain regions of $T^*_uTM$, in which case the function $\alpha(u,p)$ is purely imaginary and cannot be the value of a reasonable Hamiltonian function. This argument motivates to consider the collection $\mathcal{D}_{Tu}$ of {\it time-like momenta} over $u\in TM$ is defined by the set of co-vectors $p\in \,T^*_u TM$ such that
\begin{align}
\alpha(u,p)=\,\sum^{8N}_{i,j=1}\,\eta^{ij}(u)\,p_i \,p_j\,>0.
\label{definitionofthecone}
\end{align}

The domain of a Hamilton-Randers function is restricted to  be the topological closure of the open submanifold $\mathcal{D}_T$ of  time-like  momenta. This is indeed a cone: if $p\in \,\mathcal{D}_{Tu}$, then $\lambda\,p\in\mathcal{D}_{Tu}$ for $\lambda\in\,\mathbb{R}^+$. Also, $\mathcal{D}_{Tu}$  is the pre-image  of an open set $(0,+\infty)$ by the Randers type function $F(u,p)$, which is continuous function on the arguments. Therefore, $\mathcal{D}_{Tu}$ is an open sub-manifold of $T^*_uTM$.

The notion of pseudo-Randers space is formulated in terms of well defined geometric objects, namely, the vector field $\beta\in\,\Gamma TTM$ and the pseudo-Riemannian norm $\alpha$. Because of this reason, a metric of pseudo- Randers type is denoted by the pair $(\alpha,\beta)$.
\bigskip
\\
{\bf Notion of Hamilton-Randers space}.
Let $\beta \in \Gamma\,TTM$ be a vector field on $TM$ such that the dual condition to the Randers condition \eqref{randerscondition}, namely, the condition
\begin{align}
|\eta^*(\beta,\beta)|<\,1, \quad  \beta\in \,\Gamma TTM
\label{boundenesscondition}
\end{align}
holds good.
\begin{definicion}
A Hamilton-Randers space is a generalized Hamilton space whose  Hamiltonian function is of the form
\begin{align}
F:\mathcal{D}_T\to \mathbb{R}^+\cup \{0\},\quad
(u,p)\mapsto\, F (u,p)=\,\alpha(u,p)+\beta(u,p).
\label{corandersmetric}
\end{align}
with  $\alpha=\,\sqrt{\eta^{ij}(u)p_ip_j}$ real on $\mathcal{D}_T\,\subset T^*TM$ and where
\begin{align*}
\beta(u,p)=\,\sum^{8N}_{i=1}\,{\beta}^i(u) p_i,
\end{align*}
 such that $\beta$ is constrained by the condition \eqref{boundenesscondition}.
\label{DefinicionHR}
\end{definicion}

A Hamilton-Randers space is characterized by a triplet $(TM,F,\mathcal{D}_T)$. Such structures admits a $t$-time inversion operation. Application of the time inversion and after several formal manipulations and transformations, as discussed in section \ref{Introduction Hamilton-Randers spaces}, the effective Hamiltonian is of the form
\begin{align}
H_t(u,p)=\,\sum^{8N}_{k=1}\,\beta^k(u)p_k.
\label{Classical Hamiltonian}
\end{align}
This Hamiltonian function defines the $U_t$ dynamics of the fundamental degrees of freedom.

 \section{Concentration of measure}
 We first start this introduction to concentration phenomena with the views on applications to Hamilton-Randers theory with a classical example of limiting law.
 \begin{ejemplo}\label{law of large numbers}{\bf Law of large numbers for independent Bernoulli's trials}.
Consider an independent sequence of Bernoulli's random variables $\{\epsilon_i,\,i=1,...,N\}$. The probability distributions are $P(1)=\,P(-1)=1/2$. Let $B_N$ the number of one in the sequence. Then elementary considerations leads to the inequality
\begin{align}
P\left(\left| B_N-\,\frac{N}{2}\right|\geq x \right)\leq 2\,\exp\left(-\,\frac{2\,x^2}{N}\right).
\end{align}
It turns out that, this inequality is not strong for $N$ not large. That is, the probability $P\left(\left| B_N-\,\frac{N}{2}\right|\geq x \right)$ does not satisfies a strong bound for arbitrary large $N$.
Let us consider the function $(\epsilon_1,...,\epsilon_N)\mapsto X=\,\sum^N_{i=1}\,\epsilon_i -\,N/2.$ The function is Lipschitz.
\end{ejemplo}
 The two elements of the above example that imply strong bounds on dispersion around median are:
 \begin{itemize}

\item The $1$-Lipschitz condition  holds (depending of the space, that a nice enough regularity condition for $f$ must hold, as in the case of Levy's isoperimetric inequality \cite{MilmanSchechtman2001}). For some spaces, other lighter regularity conditions imply also concentration (see for instance Levy's lemma for continuous functions on spheres \cite{MilmanSchechtman2001}). In general, some short of regularity for the functions that concentrate is required.

 \item The condition of large $N$ (large dimensional spaces) must hold. Although concentration inequalities hold, only if $N$ is large, the almost constant character of the function  on the given space will hold good.

 \end{itemize}
 The realization of this phenomenon is the basis for the theory of concentration in asymptotic analysis, geometry and probability theory \cite{MilmanSchechtman2001,Gromov,Talagrand}.

 Let us discuss the application of concentration of measure (see for instance the theory as developed in \cite{Gromov, MilmanSchechtman2001} and also the introduction from \cite{Talagrand}) to Hamilton-Randers dynamical models. The concentration of measure is a general property of regular enough functions defined in high dimensional topological spaces  ${\bf T}$ endowed with a metric function $d:{\bf T}\times {\bf T}\to\mathbb{R}$ and a Borel measure $\mu_P$ of finite measure, $\mu_P({\bf T})<\,+\infty$ or a $\sigma$-finite measure spaces (countable union of finite measure spaces). For our applications, we also require that the topological space {\bf T} has associated a local dimension. This will be the case, since we shall have that ${\bf T}\cong T^*TM$, which are locally homeomorphic to $\mathbb{R}^{16\,N}$, if $M_4$ is the four dimensional spacetime and $N$ is the number of sub-quantum molecules determining the Hamilton-Randers system. The phenomenon of concentration of measure for the category of topological spaces  admitting a well defined dimension, can be stated as  follows \cite{Talagrand},
 \bigskip
\\
 {\it In a  measure metric space of large dimension, every real $1$-Lipschitz function of many variables is almost constant almost everywhere.}
\bigskip
\\
Here the notion of {\it dimension} needs to be specified, since the space in question is not necessarily of the form $\mathbb{R}^{16N}$, but for our applications the notion of dimension is the usual one of differential geometry.

In the formalization of the concept of concentration of measure one makes use of the metric and measure structures of the space {\bf T} to provide a  precise meaning for the notions of {\it almost constant} and {\it almost everywhere}. The notions of measure structure $\mu_P$ and metric structure $d:{\bf T}\times {\bf T}\to \mathbb{R}$ are independent from each other. Indeed, the standard framework where  concentration is formulated is in the  category of mm-spaces \cite{Gromov} and for $1$-Lipschitz functions. However, for the reasons discussed above, we shall pay attention to a class of such spaces, namely, the one admitting a well defined local dimension. Despite this, the spaces that we shall consider will be called $mm$-spaces and denoted as for example $({\bf T},\mu_P,d)$.

In a measure metric space $({\bf T},\mu_P, d)$, the {\it concentration function}
\begin{align*}
\alpha(\mu_P):\mathbb{R}\to \mathbb{R},\quad \rho\mapsto \alpha(\mu_P,\rho)
\end{align*}
is defined by the condition that $\alpha(\mu_P,\rho)$ is the smallest real number such that
\begin{align}
\mu_P(|f-M_f|>\rho)\leq 2\,\alpha(\mu_P,\rho),
\label{concentationformula}
\end{align}
for any $1$-Lipschitz function $f:{\bf T}\to \mathbb{R}$. Thus $\alpha(\mu_P,\rho)$ does not depend on the function $f$. $M_f$ is the {\it median} or  {\it Levy's mean} of $f$, which is defined as the value attained by $f:{\bf T}\to \mathbb{R}$ such that
\begin{align*}
\mu_P(f>M_f)=1/2\,\textrm{ and }\, \mu_P(f<M_f)=1/2.
\end{align*}
Therefore,  the {\it probability} that the function $f$  differs from the median $M_f$ in the sense of the given measure $\mu_P$ by more than the given value $\rho\in \,\mathbb{R}$ is bounded by the concentration function $\alpha(\mu_P,\rho)$. Note that the function $f$ must be conveniently normalized, in order to differences be compared with real numbers.

\begin{ejemplo}
A relevant example of concentration of measure is provided by the concentration of measure in spheres $\mathbb{S}^N\subset \, \mathbb{R}^{N+1}.$ Let $(\mathbb{S}^N,\mu_S, d_S)$ be the $N$-dimensional sphere with the standard measure and the round metric distance function. As a consequence of the isoperimetric inequality \cite{MilmanSchechtman2001} it holds that for each $A\subset \,\mathbb{S}^N$ with $\mu(A)\geq 1/2$ and  $\epsilon\in \,(0,1)$, the set
\begin{align*}
A_{\epsilon}:=\{x\in\,\mathbb{S}^N\,s.t.\,d_S(x,A)\leq\epsilon\}
 \end{align*}
 is such that
 \begin{align}
 \mu_{{S}}(A_{\epsilon})\geq 1-\sqrt{\pi/8}\exp(-\epsilon^2\,(N-1)/2).
  \label{concentrationforspheres}
  \end{align}
  Let $f:\mathbb{S}^N\to \mathbb{R}$ be a $1$-Lipschitz function and let us consider the set
\begin{align*}
 A:=\{x\in\,\mathbb{S}^N\,s.t.\,f(x)\leq\,M_f\}.
 \end{align*}
 Then from the definition of the Levy's mean, it turns out that $\mu_S(A)\geq 1/2$, which leads to the
 the concentration inequality \cite{MilmanSchechtman2001}
 \begin{align*}
 \mu_{{S}}(A_{\epsilon})\geq \,1-\sqrt{\pi/8}\,\exp(-\epsilon^2\,(N-1)/2).
 \end{align*}

  An important consequence of the relation \eqref{concentrationforspheres} is that for high dimensional spheres  $N\to +\infty$ and for each $\epsilon\in\,(0,1)$ (note that the radius of the sphere is normalized, such that $\epsilon =1$ is the maximal distance between points in the sphere), for almost  all the points on the sphere (that is, module a set of measure zero by the measure $\mu_S$) the limit condition
  \begin{align*}
  \lim_{N\to \infty}\,\mu(A_{\epsilon})\to 1
  \end{align*}
  holds good.
  That is, the $1$-Lipschitz function $f$ must be almost constant on $\mathbb{S}^{N}$, concentrating its value at the median $M_f$. In particular, for the sphere $\mathbb{S}^N$
the concentration of $1$-Lipschitz functions is of the form
\begin{align}
\alpha(P_M,\rho)\leq \, C \exp\left(-\frac{(N-1)}{2}\,\rho^2\right),
\label{concentration1}
\end{align}
with $C$ a constant of order $1$.
\label{Example concentration 1}
\end{ejemplo}
\begin{ejemplo}
The second example of concentration that we consider here refers to  $1$-Lipschitz real functions on $\mathbb{R}^N$
 (compare with \cite{Talagrand}, pg. 8). In this case,  the concentration inequality is of the form
\begin{align}
\mu_P\left(\left|f-M_f\right|\,\frac{1}{\sigma_f}\,>\frac{\rho}{\rho_P}\right)\leq \, \frac{1}{2} \exp\left(-\frac{\rho^2}{2\rho^2_P}\right),
\label{concentration2}
\end{align}
where we have adapted the example from \cite{Talagrand} to a Gaussian measure $\mu_P$ with median $M_f$. In the application of this concentration to Hamilton-Randers models, $\rho_P$ is a measure of the minimal standard contribution to the distance $\rho$ per unit of degree of freedom of the Hamilton-Randers system. $\frac{\rho}{\rho_P}$ must be independent of the function $f$. $\sigma_f$ is associated to the most precise physical resolution of any measurement of the quantum  observable associated to the $1$-Lipschitz function $f:\mathbb{R}^N\to \mathbb{R}$.
\label{Example concentration 2}
\end{ejemplo}

\begin{comentario}
The condition of large dimension is not enough to have the effect of concentration of measure, even for large dimensional spaces, as the inequality \eqref{concentration2} of the above Example \ref{Example concentration 2} shows (see also the equivalent inequality in (2.9) in \cite{Talagrand}).
\end{comentario}


\begin{thebibliography}{22}

 \bibitem{AcostaFernandezetIsidroSantander2013} D. Acosta, P.  Fern\'andez de C\'ordoba, J. M. Isidro and J. L. G. Santander, {\it Emergent quantum mechanics as a classical, irreversible thermodynamics}  {\it Int. J. Geom. Methods Mod. Phys.} {\bf 10}, 1350007 (2013).

\bibitem{Adler}  S. L. Adler, {\it Quantum Theory as an Emergent Phenomenon: The Statistical Mechanics of Matrix Models as the Precursor of Quantum Field Theory}, Cambridge University Press (2004).

\bibitem{Arnold} V. Arnold, {\it Mathematical Methods of Classical Mechanics}, Springer (1989).

\bibitem{BaoChernShen} D. Bao, S.S. Chern and Z. Shen, {\it An Introduction to Riemann-Finsler
Geometry}, Graduate Texts in Mathematics 200, Springer-Verlag.

 \bibitem{Bars2001} I. Bars, {\it Survey of Two-Time Physics}, Class.Quant.Grav. {\bf 18} 3113-3130 (2001).

\bibitem{Bell} J. S. Bell, {\it Speakable and Unspeakable in Quantum Mechanics}, Cambridge University Press (1987).

\bibitem{Berger2002} M. Berger {A panoramic view of Riemannian Geometry}, Springer-Verlag (2002).

\bibitem{Blasone2}M. Blasone, P. Jizba and F. Scardigli, {\it Can Quantum Mechanics be an Emergent
Phenomenon?}, J. Phys. Conf. Ser. {\bf 174} (2009) 012034, arXiv:0901.3907[quant-ph].

\bibitem{Bohm1980} D. Bohm, {\it Wholeness and the Implicate Order} , London: Routledge,  (1980).

\bibitem{CaianielloGasperiniScarpetta} E. R. Caianiello, M. Gasperini and G. Scarpetta, {\it Inflaction and singularity
    prevention in a model for extended-object-dominated cosmology}, Clas. Quan. Grav. {\bf 8}, 659 (1991).

\bibitem{CraigWeinstein} W. Craig and S. Weinstein, {\it On determinism and well-possedness in multiple time dimensions}, Proc. Royal. Soc. London, A: Mathematical, Physical and Engineering Science {\bf 465}, 3023–3046 (2009).

\bibitem{Diosi} L. Di\'osi, {\it Gravitation and quantum-mechanical localization of macro-objects}, Phys. Lett. A {\bf 105}, 199; L. Di\'osi  {\it Phys. Rev. A} {\bf 40} 1165, (1989).


\bibitem{FernandezIsidroPerea2013} P. Fern\'andez de C\'ordoba, J. M. Isidro and Milton H. Perea {\it Emergent Quantum Mechanics as thermal  essemble}, Int. J. Geom. Meth. Mod. Phys. {\bf 11}, 1450068 (2014).

 \bibitem{Ricardo06}R. Gallego Torrom\'e, {\it A finslerian version of 't Hooft Deterministic Quantum Models},
J. Math. Phys. {\bf  47}, 072101 (2006).

\bibitem{Ricardo2015} R.Gallego Torrom\'e, {\it Emergence of classical gravity and the objective reduction
of the quantum state in deterministic models of quantum mechanics},  Journal of Physics:
Conference Series {\bf 626} 1, 012073 (2015).

\bibitem{Ricardo2014} R. Gallego Torrom\'e, {\it Foundations for a theory of emergent quantum mechanics and emergent classical gravity}, arXiv:1402.5070 [math-ph].

\bibitem{Ricardo2017} R. Gallego Tottom\'e, {\it Emergent quantum mechanics and the origin of quantum non-local correlations}, International Journal of Theoretical Physics volume {\bf 56}, 3323 (2017)

\bibitem{Ricardo2020} R. Gallego Torrom\'e, {\it Maximal acceleration geometries and spacetime curvature bounds}, arXiv:1907.00781, to appear in {\it Int. J. Geom. Methods Mod. Phys.}.

\bibitem{GhirardiGrassiRimini} G. Ghirardi, R. Grassi and A. Rimini, {\it Continuous-spontaneous-reduction model involving gravity}, Phys. Rev. A {\bf 42} 1057 (1990).

\bibitem{Gromov} M. Gromov, {\it Riemannian structures for Riemannian and non-Riemannian spaces}, Birkh$\ddot{a}$user (1999).


\bibitem{Groessing2013} G. Gr\"{o}essing, {\it Emergence of Quantum Mechanics from a Sub-Quantum Statistical Mechanics}, Int. J. Mod. Phys. B, {\bf 28}, 1450179 (2014).

\bibitem{Hooft} G. 't Hooft, {\it Determinism and Dissipation in Quantum Gravity}, Class. Quantum Grav. {\bf 16}, 3263 (1999).


\bibitem{Hooft2006} G. 't Hooft, {\it Emergent Quantum Mechanics and Emergent Symmetries}, 13th International Symposium on Particles, Strings, and Cosmology-PASCOS 2007. AIP Conference Proceedings {\bf 957}, pp. 154-163 (2007).

\bibitem{Hooft2006b} G. 't Hooft, {\it The mathematical basis for deterministic quantum mechanics
}, ITP-UU-06/14, SPIN-06/12,{ quant-ph/0604008}.


\bibitem{Hooft2016} G. 't Hooft, {\it The Cellular Automaton Interpretation of Quantum Mechanics}, Fundamental Theories in Physics Vol. {\bf 185}, Springer Verlag (2016).

\bibitem{Koopman1931} B. O. Koopman, {\it Hamiltonian Systems and Transformations in Hilbert Space} Proceedings of the National Academy of Sciences {\bf 17} (5), 315 (1931).

\bibitem{MilmanSchechtman2001} V. D. Milman and G. Schechtman, {\it Asymptotic theory of Finite Dimensional normed spaces}, Lecture notes in Mathematics 1200, Springer (2001).

\bibitem{Penrose} R. Penrose, {\it On Gravity's Role in Quantum State Reduction}, Gen. Rel. and Gravit. {\bf 8}, No. 5, 581 (1996); R. Penrose, {\it The Road to Reality}, Vintage, London (2005).

\bibitem{Randers} G. Randers, {\it On an Asymmetrical Metric in the
Four-Space of General Relativity}, Phys. Rev. {\bf 59},
195-199 (1941).

\bibitem{ReedSimonI} M. Reed M and B. Simon, {\it Methods of Modern Mathematical Physics I: Functional Analysis, Revised and Enlarged Edition}, New York: Academic Press (1980).

\bibitem{Singh2019 a} T. P. Singh  {\it  Proposal for a new quantum theory of gravity},
	Z. Naturforsch. A {\bf 74}, 617 (2019).

\bibitem{Talagrand} M. Talagrand, {\it A new look at independence}, Ann. Probab. {\bf 24}, Number 1, 1-34 (1996).

\bibitem{von Neumann} J. von Neumann,  {\it Zur Operatorenmethode In Der Klassischen Mechanik}, Annals of Mathematics {\bf 33} (3): 587–642 (1932);
J. von Neumann,  {\it Zusatze Zur Arbeit "Zur Operatorenmethode.... }, Annals of Mathematics {\bf 33} (4): 789–791 (1932).


\bibitem{Whitehead1932} J. H. C. Whitehead, {\it Convex regions in the geometry of paths}, Quarterly Journal of Mathematics - Quart. J. Math. {\bf 3}, no. 1, pp. 33-42 (1932).

\end{thebibliography}
\end{document}